\documentclass[11pt]{article}

\usepackage{wrapfig}
\usepackage{times}
\usepackage{citesort}
\usepackage{graphicx}
\usepackage{color}
\usepackage{fullpage}
\usepackage{amssymb}
\usepackage{amsmath}

\newcommand{\eps}{\epsilon}

\newtheorem{lemma}{Lemma}
\newtheorem{theorem}{Theorem}

\renewcommand{\H}{\mathcal{H}}

\newcommand{\F}{\mathcal{F}}

\title{Dynamic External Hashing: The Limit of Buffering}

\author{Zhewei Wei \and Ke Yi \and Qin Zhang}

\date{Hong Kong University of Science and Technology \\
Clear Water Bay, Hong Kong, China\\
$\{$wzxac, yike, qinzhang$\}$@cse.ust.hk}

\newenvironment{proof}{\trivlist\item[]\emph{Proof}:}%
{\unskip\nobreak\hskip 1em plus 1fil\nobreak$\Box$
\parfillskip=0pt%
\endtrivlist}

\begin{document}
\maketitle

\begin{abstract}
  Hash tables are one of the most fundamental data structures in computer
  science, in both theory and practice.  They are especially useful in
  external memory, where their query performance approaches the ideal cost
  of just one disk access.  Knuth \cite{k-ss-73} gave an elegant analysis
  showing that with some simple collision resolution strategies such as
  linear probing or chaining, the expected average number of disk I/Os of a
  lookup is merely $1+1/2^{\Omega(b)}$, where each I/O can read a disk
  block containing $b$ items.  Inserting a new item into the hash table
  also costs $1+1/2^{\Omega(b)}$ I/Os, which is again almost the best one
  can do if the hash table is entirely stored on disk.  However, this
  assumption is unrealistic since any algorithm operating on an external
  hash table must have some internal memory (at least $\Omega(1)$ blocks)
  to work with.  The availability of a small internal memory buffer can
  dramatically reduce the amortized insertion cost to $o(1)$ I/Os for many
  external memory data structures.  In this paper we study the inherent
  query-insertion tradeoff of external hash tables in the presence of a
  memory buffer.  In particular, we show that for any constant $c>1$, if
  the query cost is targeted at $1+O(1/b^{c})$ I/Os, then it is not
  possible to support insertions in less than $1-O(1/b^{\frac{c-1}{4}})$
  I/Os amortized, which means that the memory buffer is essentially
  useless.  While if the query cost is relaxed to $1+O(1/b^{c})$ I/Os for
  any constant $c<1$, there is a simple dynamic hash table with $o(1)$
  insertion cost.  These results also answer the open question recently
  posed by Jensen and Pagh \cite{jensen07:_optim}.
\end{abstract}

\newpage

\section{Introduction}

Hash tables are the most efficient way of searching for a particular item
in a large database, with constant query and update times.  They are
arguably one of the most fundamental data structures in computer science,
due to their simplicity of implementation, excellent performance in
practice, and many nice theoretical properties.  They work especially well
in external memory, where the storage is divided into disk blocks, each
containing up to $b$ items.  Thus collisions happen only when there are
more than $b$ items hashed into the same location.  Using some common
collision resolution strategies such as linear probing or chaining, the
expected average cost of a successful lookup of an external hash table is
merely $1+1/2^{\Omega(b)}$ disk accesses (or simply {\em I/O}s), provided
that the {\em load factor}\footnote{The load factor is defined to be ratio
  between the minimum number of blocks required to store $n$ data records,
  $\lceil n/b \rceil$, and the actual number of blocks used by the hash
  table.}  $\alpha$ is less than a constant smaller than 1. The expectation
is with respect to the random choice of the hash function, while the
average is with respect to the uniform choice of the queried item.  An
unsuccessful lookup costs slightly more, but is the same as that of a
successful lookup if ignoring the constant in the big-Omega.  Knuth
\cite{k-ss-73} gave an elegant analysis deriving the exact formula for the
query cost, as a function of $\alpha$ and $b$.  As typical values of $b$
range from a few hundreds to a thousand, the query cost is extremely close
to just one I/O; some exact numbers are given in \cite[Section
6.4]{k-ss-73}.

Inserting or deleting an item from the hash table also costs
$1+1/2^{\Omega(b)}$ I/Os: We simply first read the target block where the
new item should go, then write it back to disk\footnote{Rigorously
  speaking, this is $2+1/2^{\Omega(b)}$ I/Os, but since disk I/Os are
  dominated by the seek time, writing a block immediately after reading it
  can be considered as one I/O. }.  If one wants to maintain the load
factor we can periodically rebuild the hash table using schemes like {\em
  extensible hashing} \cite{fagin:extendible} or {\em linear hashing}
\cite{litwin80:_linear}, but this only adds an extra cost of $O(1/b)$ I/Os
amortized.  Jensen and Pagh \cite{jensen07:_optim} demonstrate how to
maintain the load factor at $\alpha = 1-O(1/b^{\frac{1}{2}})$ while still
supporting queries in $1+O(1/b^{\frac{1}{2}})$ I/Os and updates in
$1+O(1/b^{\frac{1}{2}})$ I/Os.  Indeed, one cannot hope for lower than 1
I/O for an insertion, if the hash table must reside on disk entirely and
there is no space in main memory for buffering.  However, this assumption
is unrealistic, since an algorithm operating on an external hash table has
to have at least a constant number of blocks of internal memory to work
with.  So we must include a main memory of size $m$ in our setting to model
the problem more accurately.  In fact, this is exactly what the standard
external memory model \cite{aggarwal:input} depicts: The system has a disk
of infinite size partitioned into blocks of size $b$, and a main memory of
size $m$.  Computation can only happen in main memory, which accesses the
disk via I/Os.  Each I/O can read or write a disk block, and the complexity
is measured by the number of I/Os performed by an algorithm.  The presence
of a no-cost main memory changes the problem dramatically, since it can be
used as a buffer space to batch up insertions and write them to disk
periodically, which could significantly reduce the amortized insertion
cost.  The abundant research in the area of I/O-efficient data structures
has witnessed this phenomenon numerous times, where the insertion cost can
be typically brought down to only slightly larger than $O(1/b)$ I/Os.
Examples include the simplest structures like stacks and queues, to more
advanced ones such as the buffer tree \cite{arge:buffer} and the priority
queue \cite{fadel:external,arge:oblivious}.  Many of these results hold as
long as the buffer has just a constant number of blocks; some require a
larger buffer of $\Theta(b)$ blocks (known as the {\em tall cache}
assumption).  Please see the surveys \cite{vitter:iosurvey,arge:eefnotes}
for a complete account of the power of buffering.

Therefore the natural question is, can we (or not) lower the insertion cost
of a dynamic hash table by buffering without sacrificing its near-perfect
query performance?  Interestingly, Jensen and Pagh \cite{jensen07:_optim}
recently posed the same question, and conjectured that the insertion cost
must be $\Omega(1)$ I/Os if the query cost is required to be $O(1)$ I/Os.

\paragraph{Our results.}
In this paper, we confirm that the conjecture of Jensen and Pagh
\cite{jensen07:_optim} is basically correct but not accurate
enough. Specifically we obtain the following results.  Consider
any dynamic hash table that supports insertions in expected
amortized $t_u$ I/Os and answers a successful lookup query in
expected $t_q$ I/Os on average.  We show that if $t_q \le
1+O(1/b^{c})$ for any constant $c>1$, then we must have $t_u \ge
1-O(1/b^{\frac{c-1}{4}})$.  This is only an additive term of
$1/b^{\Omega(1)}$ away from how the standard hash table is
supporting insertions, which means that buffering is essentially
useless in this case.  However, if the query cost is relaxed to
$t_q \le 1+O(1/b^{c})$ for any constant $0<c<1$, we present a
simple dynamic hash table that supports insertions in $t_u =
O(b^{c-1}) = o(1)$ I/Os.  For this case we also present a matching
lower bound of $t_u = \Omega(b^{c-1})$.  Finally for the case $t_q
= 1+\Theta(1/b)$, we show a tight bound of $t_u=\Theta(1)$. Our
results are pictorially illustrated in Figure~\ref{fig:tradeoff},
from which we see that we now have an almost complete
understanding of the entire query-insertion tradeoff, and $t_q = 1
+ \Theta(1/b)$ seems to be the sharp boundary separating
effective and ineffective buffering.  We prove our lower bounds
for the three cases above using a unified framework in
Section~\ref{sec:lower-bound}.  The upper bound for the first case
is simply the standard hash table following \cite{k-ss-73}; we
give the upper bounds for the other two cases in
Section~\ref{sec:upper-bound}.

\begin{figure}[h]
\begin{center}
\includegraphics[width=8.5cm]{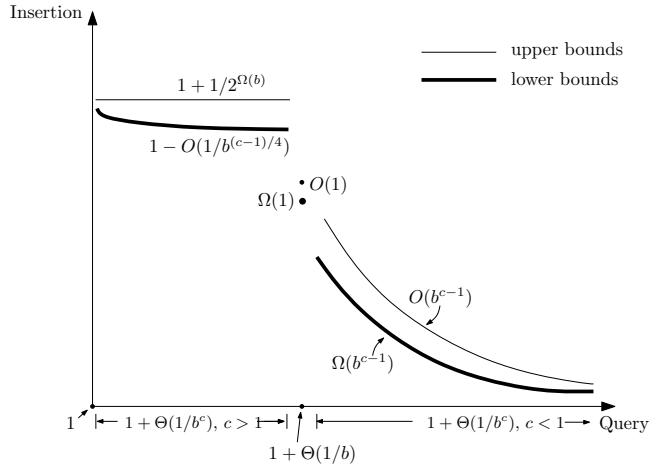}
\caption{The query-insertion tradeoff.}
\label{fig:tradeoff}
\end{center}
\end{figure}

In this paper we only consider the query-insertion tradeoff for
the following reasons.  First, our primary interest is on the
lower bound, a query-insertion tradeoff lower bound is certainly
applicable to the query-update tradeoff for more general updates
that include both insertions and deletions.  And secondly, there
tends to be a lot more insertions than deletions in many practical
situations like managing archival data.  For similar reasons we
only consider the query cost as that of a successful lookup.

Let $h(x)$ be a hash function that maps an item $x$ to a hash
value between $0$ and $u-1$.  In our lower bound construction, we
will insert a total of $n$ independent items such that each $h(x)$
is uniformly randomly distributed between $0$ and $u-1$, and we
prove a lower bound on the expected amortized cost per insertion,
under the condition that at any time, the hash table must be able
to answer a query for the already inserted items with the desired
expected average query bound.  Thus, our lower bound holds even
assuming that $h(x)$ is an ideal hash function that maps each item
to a hash value independently uniformly at random, a justifiable
assumption \cite{mitzenmachery08:_why_} often made in many works
on hashing.  Also note that since we use an input that is
uniformly at random, it is sufficient to consider only
deterministic algorithms as randomization will not help any more.

When proving our lower bound we make the only requirement that
items must be treated as atomic elements, i.e., they can only be
moved or copied between memory and disk in their entirety, and
when answering a query, the query algorithm must visit the block
(in memory or on disk) that actually contains the item or one of
its copies.  Such an {\em indivisibility} assumption is also made
in the sorting and permuting lower bounds in external memory
\cite{aggarwal:input}.  We assume that each machine word consists
of $\log u$ bits and each item occupies one machine word.  A block
has $b$ words and the memory stores up to $m$ words. We assume
that each block is not too small: $b > \log u$. Our lower and
upper bounds hold for the wide range of parameters
$\Omega\left(b^{1+2c}\right) < \frac{n}{m} < 2^{o(b)}$. Finally, we
comment that our lower bounds do not depend on the load factor,
which implies that the hash table cannot do better by consuming
more disk space.

\paragraph{Related results.}
Hash tables are widely used in practice due to their simplicity and
excellent performance.  Knuth's analysis \cite{k-ss-73} applies to the
basic version where $h(x)$ is assumed to be an ideal random hash function
and $t_q$ is the expected average cost.  Afterward, a lot of works have
been done to give better theoretical guarantees, for instance removing the
ideal hash function assumption \cite{CarWeg79}, making $t_q$ to be
worst-case \cite{fks-sstwc-84,dkmmrt-dphul-94,pagh04:_cuckoo_hashin}, etc.
Please see \cite{pagh02:_hashin} for a survey on hashing techniques.  Lower
bounds have been sparse because in internal memory, the update time cannot
be lower than $\Omega(1)$, which is already achieved by the standard hash
table.  Only with some strong requirements, e.g., when the algorithm is
deterministic and $t_q$ is worst-case, can one obtain some nontrivial lower
bounds on the update time \cite{dkmmrt-dphul-94}.  Our lower bounds, on the
other hand, hold for randomized algorithms and do not need $t_q$ to be
worst-case.

As commented earlier, in external memory there is a trivial lower bound of
1 I/O for either a query or an update, if all the changes to the hash table
must be committed to disk after each update.  However, the vast amount of
works in the area of external memory algorithms have never made such a
requirement.  And indeed for many problems, the availability of a small
internal memory buffer can significantly reduce the amortized update cost
without affecting the query cost
\cite{arge:buffer,fadel:external,arge:oblivious,arge:eefnotes,vitter:iosurvey}.
Unfortunately, little is known on the inherent limit of what buffering can
do.  The only nontrivial lower bound on the update cost of any external
data structure with a memory buffer is a paper by Fagerberg and Brodal
\cite{brodal03:_lower}, who gave a query-insertion tradeoff for the {\em
  predecessor} problem in a natural external version of the comparison
model, a model much more restrictive than the indivisibility model we use.
As assuming a comparison-based model precludes any hashing techniques,
their techniques are inapplicable to the problem we have at hand.  To the
best of our knowledge, no nontrivial lower bound on external hashing of any
kind is known.



\section{Lower Bounds}
\label{sec:lower-bound}
To obtain a query-insertion tradeoff, we start with an empty hash table and
insert a total of $n$ independent items such that $h(x)$ is uniformly
randomly distributed in $U=\{0,\dots, u-1\}$.  We will derive a lower bound
on $t_u$, the expected amortized number of I/Os for an insertion, while
assuming that the hash table is able to answer a successful query in $t_q$
I/Os on average in expectation after the first $i$ items have been
inserted, for all $i=1,\dots, n$.  We assume that all the $h(x)$'s are
different, which happens with probability $1-O(1/n)$ as long as $u>n^3$ by
the birthday paradox. In the sequel we will not distinguish between an item
$x$ and its hash value $h(x)$.  Under this setting we obtain the following
tradeoffs between $t_q$ and $t_u$.

\begin{theorem}
\label{thm:insert-query-tradeoff} For any constant $c>0$, suppose
we insert a sequence of $n> \Omega\left(m \cdot b^{1+2c}\right)$
random items into an initially empty hash table.  If the total
cost of these insertions is expected $n\cdot t_u$ I/Os, and the
hash table is able to answer a successful query in expected
average $t_q$ I/Os at any time, then the following tradeoffs hold:
\begin{enumerate}
\item If $t_q \le 1 + O(1/b^{c})$ for any $c > 1$, then $t_u \ge 1
- O(1/b^{\frac{c-1}{4}})$; 
\item If $t_q \le 1 + O(1/b)$, then $t_u \ge \Omega(1)$; 
\item If $t_q
  \le 1 + O(1/b^{c})$ for any $0< c < 1$, then $t_u \ge \Omega(b^{c-1})$.
\end{enumerate}
\end{theorem}

\paragraph{The abstraction.}
To abstractly model a dynamic hash table, we ignore any of its auxiliary
structures but only focus on the layout of items.  Consider any snapshot of
the hash table when we have inserted $k$ items.  We divide these $k$ items
into three zones.  The {\em memory zone} $M$ is a set of at most $m$ items
that are kept in memory.  It takes no I/O to query any item in $M$.  All
items not in $M$ must reside on disk.  Denote all the blocks on disk by
$B_1,B_2,\dots, B_d$.  Each $B_i$ is a set of at most $b$ items, and it is
possible that one item appears in more than one $B_i$.  Let $f:U
\rightarrow \{1,\dots,d\}$ be any function computable within memory, and we
divide the disk-resident items into two zones with respect to $f$ and the
set of blocks $B_1,\dots,B_d$.  The {\em fast zone} $F$ contains all items
$x$ such that $x \in B_{f(x)}$: These are the items that are accessible
with just one I/O.  We allocate all the remaining items into the {\em slow
  zone} $S$: These items need at least two I/Os to locate.  Note that under
random inputs, the sets $M, F, S, B_1,\dots, B_d$ are all random sets.

Any query algorithm on the hash table can be modeled as described, since
the only way to find a queried item in one I/O is to compute the index of a
block containing $x$ with only the information in memory.  If the
memory-resident computation gives an incorrect address or anything else, at
least 2 I/Os will be necessary.  Because any such $f$ must be computable
within memory, and the memory has $m\log u$ bits, the hash table can employ
a family $\F$ of at most $2^{m\log u}$ distinct $f$'s.  Note that the
current $f$ adopted by the hash table is dependent upon the already
inserted items, but the family $\F$ has to be fixed beforehand.

Suppose the hash table answers a successful query with an expected
average cost of $t_q = 1+\delta$ I/Os, where $\delta = 1/b^{c}$
for any constant $c>0$. Consider the snapshot of the hash table
when $k$ items have been inserted.   Then we must have $E[|F| + 2
\cdot |S|] / k \le 1+\delta$. Since $|F| + |S| = k - |M|$ and
$E[|M|] \le m$, we have
\begin{equation}
\label{eq:1} E[|S|] \le m + \delta k.
\end{equation}

We also have the following high-probability version of \eqref{eq:1}.
\begin{lemma}
\label{lem:small-query} Let $\phi \ge 1/b^{(c-1)/4}$ and let $k\ge
\phi n$. At the snapshot when $k$ items have been inserted, with
probability at least $1 - 2\phi$, $|S| \le m + \frac{\delta}{\phi}k$.
\end{lemma}
\begin{proof}
  On this snapshot the hash table answers a query in
  expected average $1+\delta$ I/Os.  We claim that with probability at most
  $2\phi$, the average query cost is more than $1+\delta/\phi$.  Otherwise,
  since in any case the average query cost is at least $1-m/k$ (assuming
  all items not in memory need just one I/O), we would have an expected
  average cost of at least
\[
(1 - 2\phi)(1 - m/k) + 2\phi \cdot (1 + \delta / \phi)
> 1 + \delta,
\]
provided that $\frac{n}{m} > \frac{1}{\phi \delta}$, which
is valid since we assume that $\frac{n}{m} > b^{1 + 2c}$.  The
lemma then follows from the same argument used to derive (\ref{eq:1}).
\end{proof}

\paragraph{Basic idea of the lower bound proof.}
For the first $\phi n$ items, we ignore the cost of their insertions.
Consider any $f: U \rightarrow \{1,\dots, d\}$.  For $i=1,\dots,d$, let
$\alpha_i = |f^{-1}(i)|/u$, and we call $(\alpha_1,\dots,\alpha_d)$ the
{\em characteristic vector} of $f$.  Note that $\sum_i \alpha_i = 1$. For
any one of the first $\phi n$ items, since it is randomly chosen from $U$,
$f$ will direct it to $B_i$ with probability $\alpha_i$. Intuitively, if
$\alpha_i$ is large, too many items will be directed to $B_i$.  Since $B_i$
contains at most $b$ items, the extra items will have to be pushed to the
slow zone.  If there are too many large $\alpha_i$'s, $S$ will be large
enough to violate the query requirement.  Thus, the hash table should use
an $f$ that distributes items relatively evenly to the blocks.  However, if
$f$ evenly distributes the first $\phi n$ items, it is also likely to
distribute newly inserted items evenly, leading to a high insertion cost.
Below we formalize this intuition.

\medskip For the first tradeoff of Theorem~\ref{thm:insert-query-tradeoff},
we set $\delta = 1/b^{c}$. We also pick the following set of parameters
$\phi = 1/b^{(c-1)/4}, \rho = 2b^{(c+3)/4}/n, s = n/b^{(c+1)/2}$.  We will
use different values for these parameters when proving the other two
tradeoffs. Given an $f$ with characteristic vector $(\alpha_1, \ldots,
\alpha_d)$, let $D^f = \{i\ |\ \alpha_i > \rho\}$ be the collection of
block indices with large $\alpha_i$'s.  We say that the indices in $D^f$
form the {\em bad index area} and others form the {\em good index area}.
Let $\lambda_f = \sum_{i \in D^f} \alpha_i$. Note that there are at most
$\lambda_f/\rho$ indices in the bad index area. We call an $f$ with
$\lambda_f > \phi$ a {\em bad function}; otherwise it is a {\em good
  function}. The following lemma shows that with high probability, the hash
table should use a good function $f$ from $\F$.

\begin{lemma}
\label{lem:discard-bad-mapping} At the snapshot when $k$ 
items are inserted for any $k \ge \phi n$, the function $f$ used by the
hash table is a good function with probability at least
$1-2\phi-1/2^{\Omega(b)}$. 
\end{lemma}
\begin{proof}
  Consider any bad function $f$ from $\F$.  Let $X_j$ be the indicator
  variable of the event that the $j$-th inserted item is mapped to the bad
  index area, $j=1,\dots,k$.  Then $X = \sum_{j=1}^k X_j$ is the total
  number of items mapped to the bad index area of $f$.  We have $E[X] =
  \lambda_f k$.  By Chernoff inequality, we have
\[\Pr\left[X < \frac{2}{3}\lambda_f k \right] \le e^{-\frac{(1/3)^2 \lambda_f
    k}{2}} \le e^{-\frac{\phi^2 n}{18}},\] namely with probability at
least $1 - e^{-\frac{\phi^2 n}{18}}$, we have $X \ge
\frac{2}{3}\lambda_f k$.  Since the family $\F$ contains at most $2^{m\log
  u}$ bad functions, by union bound we know that with probability at least
$1 - 2^{m\log u} \cdot e^{-\frac{\phi^2 n}{18}} \ge 1 - 1/2^{\Omega(b)}$
(by the parameters chosen and the assumption that $n > \Omega(m
b^{1+2c}),b>\log u$), for all the bad functions in $\F$, we have $X \ge
\frac{2}{3}\lambda_f k$.

Consequently, since the bad index area can only accommodate $b \cdot
\lambda_f/\rho$ items in the fast zone, at least $\frac{2}{3}\lambda_f k -
b\lambda_f/\rho$ cannot be in the fast zone.  The memory zone can accept at
most $m$ items, so the number of items in the slow zone is at least
\[ |S| \ge \frac{2}{3}\lambda_f k - b\lambda_f/\rho - m > m +
\frac{\delta}{\phi} k.
\]
This happens with probability at least $1 - 1/2^{\Omega(b)}$, due to the
fact that $f$ is a bad function.  On the other hand,
Lemma~\ref{lem:small-query} states that $|S| \le m + \frac{\delta}{\phi}k$
holds with probability at least $1-2\phi$, thus by union bound $f$ is a
good function with probability at least $1-2\phi - 1/2^{\Omega(b)}$.
\end{proof}

\paragraph{A bin-ball game.}
Lemma~\ref{lem:discard-bad-mapping} enables us to consider only
those good functions $f$ after the initial $\phi n$
insertions. To show that any good function will incur a
large insertion cost, we first consider the following bin-ball
game, which captures the essence of performing insertions using a
good function.

In an {\em $(s, p, t)$ bin-ball game}, we throw $s$ balls into $r$ (for any
$r \ge 1/p$) bins independently at random, and the probability that any
ball goes to any particular bin is no more than $p$. At the end of the
game, an adversary removes $t$ balls from the bins such that the remaining
$s-t$ balls hit the least number of bins. The cost of the game is defined
as the number of nonempty bins occupied by the $s-t$ remaining balls.

We have the following two results with respect to such a game, depending on
the relationships among $s,p$, and $t$.

\begin{lemma}
\label{lem:binball-game} If $sp \le \frac{1}{3}$, then for any $\mu
> 0$, with probability at least $1 - e^{-\frac{\mu^2 s}{3}}$, the cost of
an $(s, p, t)$ bin-ball game is at least $(1 - \mu)(1 - sp)s - t$.
\end{lemma}

\begin{proof}
  Imagine that we throw the $s$ balls one by one.  Let $X_j$ be the
  indicator variable denoting the event that the $j$-th ball is thrown into
  an empty bin.  The number of nonempty bins in the end is thus $X =
  \sum_{j = 1}^{s}X_j$. These $X_j$'s are not independent, but no matter
  what has happened previously for the first $j-1$ balls, we always have
  $\Pr[X_j=0] \le sp$.  This is because at any time, at most $s$ bins are
  nonempty.  Let $Y_j\ (1 \le j \le s)$ be a set of independent variables
  such that
\[ Y_i = \left\{
  \begin{array}{ll}
    0, & \textrm{with probability }
sp; \\
    1, & \textrm{otherwise.}
\end{array} \right.
\]
Let $Y = \sum_{j=1}^s Y_j$.  Each $Y_i$ is stochastically dominated by
$X_i$, so $Y$ is stochastically dominated by $X$.  We have $E[Y] = (1 -
sp)s$ and we can apply Chernoff inequality on $Y$:
$$\Pr\left[ Y < (1 - \mu)(1 - sp)s \right] < e^{-\frac{\mu^2(1 - sp)s}{2}}
< e^{-\frac{\mu^2 s}{3}}.$$
Therefore with probability at least $1 -
e^{-\frac{\mu^2 s}{3}}$, we have $X \ge (1 - \mu)(1 - sp)s$. Finally, since
removing $t$ balls will reduce the number of nonempty bins by at most $t$,
the cost of the bin-ball game is at least $(1 - \mu)(1 - sp)s - t$ with
probability at least $1 - e^{-\frac{\mu^2 s}{3}}$.
\end{proof}

\begin{lemma}
\label{lem:binball-game-2} If $s/2 \ge t$ and $s/2 \ge 1/p$, then
with probability at least $1 - 1/2^{\Omega(s)}$, the cost of an
$(s, p, t)$ bin-ball game is at least $1/(20p)$.
\end{lemma}

\begin{proof}
  In this case, the adversary will remove at most $s/2$ balls in the end.
  Thus we show that with very small probability, there exist a subset of
  $s/2$ balls all of which are thrown into a subset of at most $1/(20p)$
  bins.  Before the analysis, we merge bins such that the probability that
  any ball goes to any particular bin is between $p/2$ and $p$, and
  consequently, the number of of bins would be between $1/p$ to $2/p$. Note
  that such an operation will only make the cost of the bin-ball game
  smaller. Now this probability is at most
  $$\sum_{i=1}^{1/(20p)}\left({2/p \choose i} {s \choose s/2} \left(
      \frac{i}{1/p} \right)^{s/2} \right) \le 2 {2/p \choose 1/(20p)} {s
    \choose s/2} \left( \frac{1/(20p)}{1/p} \right)^{s/2} \le
  1/2^{\Omega(s)},$$
  hence the lemma.
\end{proof}

Now we are ready to prove the main theorem.

\begin{proof} (of Theorem~\ref{thm:insert-query-tradeoff})
  We begin with the first tradeoff. Recall that we choose the following
  parameters: $\delta = 1/b^{c}$, $\phi = 1/b^{(c-1)/4}, \rho =
  2b^{(c+3)/4}/n, s = n/b^{(c+1)/2}$. For the first $\phi n$ items, we do
  not count their insertion costs.  We divide the rest of the insertions
  into rounds, with each round containing $s$ items.  We now bound the
  expected cost of each round.
  
  Focus on a particular round, and let $f$ be the function used by the hash
  table at the end of this round.  We only consider the set $R$ of items
  inserted in this round that are mapped to the good index area of $f$,
  i.e., $R=\{ x \mid f(x) \not\in D^f\}$; other items are assumed to have
  been inserted for free.  Consider the block with index $f(x)$ for a
  particular $x$.  If $x$ is in the fast zone, the block $B_{f(x)}$ must
  contain $x$.  Thus, the number of distinct indices $f(x)$ for $x \in R
  \cap F$ is an obvious lower bound on the I/O cost of this round.  Denote
  this number by $Z = |\{f(x) \mid x \in R \cap F\}|$.  Below we will show
  that $Z$ is large with high probability.
  
  We first argue that at the end of this round, each of the following three
  events happens with high probability.
\begin{itemize}
\item $\mathcal{E}_1$: $|S| \le \delta n/\phi + m$;
\item $\mathcal{E}_2$: $f$ is a good function;
\item $\mathcal{E}_3$: For all good function $f \in \F$ and corresponding slow zones $S$ and
  memory zones $M$, $Z \ge (1 - O(\phi))s
  - t$, where $t=|S|+|M|$.
\end{itemize}

By Lemma~\ref{lem:small-query}, $\mathcal{E}_1$ happens with probability at
least $1 - 2\phi$.  By Lemma~~\ref{lem:discard-bad-mapping},
$\mathcal{E}_2$ happens with probability at least $1-2\phi -
1/2^{\Omega(b)}$.  It remains to show that $\mathcal{E}_3$ also happens
with high probability.

We prove so by first claiming that for a particular $f \in \F$ with probability at least $1 - e^{- 2\phi^2
  s}$, $Z$ is at least the cost of a $((1 - 2\phi)s, \frac{\rho}{1 -
  \lambda_f}, t)$ bin-ball game, for the following reasons:
\begin{enumerate}
\item Since $f$ is a good function, by Chernoff inequality, with
  probability at least $1 - e^{- 2 \phi^2 s}$, more than $(1 - 2\phi)s$
  newly inserted items will fall into the good index area of $f$, i.e.,
  $|R| > (1-2\phi)s$.

\item The probability of any item being mapped to any index in the good
  index area, conditioned on that it goes to the good index area, is no
  more than $\frac{\rho}{1 - \lambda_f}$.

\item Only $t$ items in $R$ are not in the fast zone $F$, excluding them
  from $R$ corresponds to discarding $t$ balls at the end of the bin-ball
  game.
\end{enumerate}

Thus by Lemma~\ref{lem:binball-game} (setting $\mu = \phi$), with
probability at least $1 - e^{- \frac{\phi^2 \cdot (1 - 2\phi)s}{3}} - e^{-
  2\phi^2 s}$, we have
\begin{eqnarray*}
Z &\ge&
(1 - \phi) \left(1 - (1 - 2\phi)s \cdot \frac{\rho}{1 - \lambda_f}
\right)(1 - 2\phi)s - t \\
&\ge& (1 - \phi) \left(1 - (1 - 2\phi)s \cdot \frac{\rho}{1 - \phi}
\right)(1 - 2\phi)s - t \ge \left(1 - O\left(
\phi\right)\right)s - t.
\end{eqnarray*}
Thus $\mathcal{E}_3$ happens  with probability at
least $1 - (e^{- \frac{\phi^2 \cdot (1 - 2\phi)s}{3}} + e^{-2\phi^2
  s})\cdot 2^{m\log u} = 1 - 2^{-\Omega(b)}$ (by the assumption that $n >
\Omega(m b^{1+2c})$ and $b > \log u$) by applying union bound on all good functions in $\F$.

Now we lower bound the expected insertion cost of one round.  By union
bound, with probability at least $1 - O(\phi) - 1/2^{\Omega(b)}$, all of
$\mathcal{E}_1, \mathcal{E}_2$, and $\mathcal{E}_3$ happen at the end of
the round.  By $\mathcal{E}_2$ and $\mathcal{E}_3$, we have $Z \ge \left(1
  - O\left(\phi\right)\right) s -t$.  Since now $t = |S| + |M|
\le \delta n/\phi + 2m = O\left( \phi s \right)$ by
$\mathcal{E}_1$, we have $Z \ge \left(1 -
  O\left(\phi\right)\right) s$.  Thus the expected cost of one
round will be at least
\[
\left(1 - O\left(\phi\right)\right)s \cdot \left(1 -
O(\phi) - 1/2^{\Omega(b)}\right) = \left(1 -
O\left(\phi\right)\right)s.
\]

Finally, since there are $(1 - \phi)n/s$ rounds, the expected amortized
cost per insertion is at least
$$\left(1 - O(\phi)\right)s \cdot (1 - \phi)n/s
\cdot 1/n = 1 - O\left(1/b^{\frac{c-1}{4}}\right).$$

\medskip For the second tradeoff, we choose the following set of
parameters: $\phi = 1/\kappa, \rho = 2\kappa b/n, s =
n/(\kappa^2b)$ and $\delta = 1/(\kappa^4 b)$ (for some constant $\kappa$
large enough). We can check that Lemma~\ref{lem:discard-bad-mapping} still
holds with these parameters, and then go through the proof above. We omit
the tedious details. Plugging the new parameters into the derivations we
obtain a lower bound $t_u \ge \Omega(1)$.

\medskip For the third tradeoff, we choose the following set of parameters:
$\phi = 1/8, \rho = 16b/n, s = 32n/b^{c}$ and $\delta = 1/b^{c}$.  We can
still check the validity of Lemma~\ref{lem:discard-bad-mapping}, and go
through the whole proof.  The only difference is that we need to use
Lemma~\ref{lem:binball-game-2} in place of Lemma~\ref{lem:binball-game},
the reason being that for our new set of parameters, we have $s\rho =
\omega(1)$ thus Lemma~\ref{lem:binball-game} does not apply. By using
Lemma~\ref{lem:binball-game-2} we can lower bound the expected insertion
cost of each round by $\Omega\left((1 - 2\phi) / (20\rho)\right)$, so the
expected amortized insertion cost is at least
$$\Omega\left(\frac{1 - 2\phi}{20\rho}\right) \cdot (1 - \phi)n/s \cdot 1/n
= \Omega(b^{c-1}),$$ as claimed.
\end{proof}


\section{Upper Bounds}
\label{sec:upper-bound}
In this section, we present a simple dynamic hash table that supports
insertions in $t_u = O(b^{c-1}) = o(1)$ I/Os amortized, while being able to
answer a query in expected $t_q = 1+O(1/b^{c})$ I/Os on average for any
constant $c<1$, under the mild assumption that $\log\frac{n}{m} = o(b)$.
Below we first state a folklore result by applying the logarithmic method
\cite{bentley:decomposable} to a standard hash table, achieving $t_u =
o(1)$ but with $t_q = \Omega(1)$.  Then we show how to improve the query
cost to $1+O(1/b^{c})$ while keeping the insertion cost at $o(1)$.  We also
show how to tune the parameters such that $t_u = \eps$ while $t_q=
1+O(1/b)$, for any constant $\eps>0$.

\paragraph{Applying the logarithmic method.}
Fix a parameter $\gamma\ge 2$.  We maintain a series of hash tables $\H_0,
\H_1, \dots$.  The hash table $\H_k$ has $\gamma^k \cdot \frac{m}{b}$
buckets and stores up to $\frac{1}{2}\gamma^{k} m$ items, so that its load
factor is always at most $\frac{1}{2}$.  It uses the $\log(\gamma^k
\cdot\frac{m}{b}) = k\log\gamma +\log \frac{m}{b}$ least significant bits
of the hash function $h(x)$ to assign items into buckets.  We use some
standard method to resolve collisions, such as chaining.  The first hash
table $\H_0$ always resides in memory while the rest stay on disk.

When a new item is inserted, it always goes to the memory-resident $\H_0$.
When $\H_0$ is full (i.e., having $\frac{1}{2}m$ items), we migrate all
items stored in $\H_0$ to $\H_1$.  If $\H_1$ is not empty, we simply merge
the corresponding buckets.  Note that each bucket in $\H_0$ corresponds to
$\gamma$ consecutive buckets in $\H_1$, and we can easily distribute the
items to their new buckets in $\H_1$ by looking at $\log \gamma$ more bits
of their hash values.  Thus we can conduct the merge by scanning the two
tables in parallel, costing $O(\gamma \cdot \frac{m}{b})$ I/Os at most.
This operation takes place inductively: Whenever $\H_k$ is full, we migrate
its items to $\H_{k+1}$, costing $O(\gamma^{k+1}\cdot\frac{m}{b})$ I/Os.
Then standard analysis shows that for $n$ insertions, the total cost is
$O(\frac{\gamma n}{b}\log\frac{n}{m})$ I/Os, or
$O(\frac{\gamma}{b}\log\frac{n}{m})$ amortized I/Os per insertion.
However, for a query we need to examine all the $O(\log_\gamma
\frac{n}{m})$ hash tables.

\begin{lemma}
\label{lem:log}
  For any parameter $\gamma\ge 2$, there is a dynamic hash table that
  supports an insertion in amortized $O(\frac{\gamma}{b}\log\frac{n}{m})$
  I/Os and a (successful or unsuccessful) lookup in expected average
  $O(\log_\gamma\frac{n}{m})$ I/Os.
\end{lemma}

\paragraph{Improving the query cost.}
Next we show how to improve the average cost of a successful query to $1 +
O(1/b^{c})$ I/Os for any constant $c<1$, while keeping the insertion cost
at $o(1)$.  The idea is to try to put the majority of the items into one
single big hash table. In the standard logarithmic method described above,
the last table may seem a good candidate, but sometimes it may only contain
a constant fraction of all items.  Below we show how to bootstrap the
structure above to obtain a better query bound.

Fix a parameter $2\le \beta \le b$.  For the first $m$ items inserted, we
dump them in a hash table $\widehat{\H}$ on disk.  Then run the algorithm
above for the next $m/\beta$ items.  After that we merge these $m/\beta$
items into $\widehat{\H}$.  We keep doing so until the size of
$\widehat{\H}$ has reached $2m$, and then we start the next round.
Generally, in the $i$-th round, the size of $\widehat{\H}$ goes from
$2^{i-1}m$ to $2^i m$, and we apply the algorithm above for every
$2^{i-1}m/\beta$ items.  It is clear that $\widehat{\H}$ always have at
least a fraction of $1-\frac{1}{\beta}$ of all the items inserted so far,
while the series of hash tables used in the logarithmic method maintain at
least a separation factor of 2 in the sizes between successive tables.
Thus, the expected average query cost is at most
\[ \left(1+1/2^{\Omega(b)}\right)\left(1\cdot\left(1-\frac{1}{\beta}\right)
  + \frac{1}{\beta}\left(2\cdot \frac{1}{2} + 3 \cdot \frac{1}{4} +
  \cdots\right)\right) = 1 + O(1/\beta).
\]

Next we analyze the amortized insertion cost.  Since the number of items
doubles every round, it is (asymptotically) sufficient to analyze the last
round.  In the last round, $\widehat{\H}$ is scanned $\beta$ times, and we
charge $O(\beta/b)$ I/Os to each of the $n$ items. The logarithmic method
is invoked $\beta$ times, but every invocation handles $O(n/\beta)$
different items.  From Lemma~\ref{lem:log}, the amortized cost per item is
still $O(\frac{\gamma}{b}\log\frac{n}{m})$ I/Os.  So the total amortized
cost per insertion is $O(\frac{1}{b}(\beta+\gamma \log\frac{n}{m}))$ I/Os.
Let the constant in this big-O be $c'$.  Then setting $\beta = b^c$ (or
respectively $\beta = \frac{\eps}{2c'}\cdot b$) and $\gamma = 2$ yields the
desired results, as long as $\log\frac{n}{m} = o(b)$.

\begin{theorem}
  For any constant $c<1, \eps>0$, there is a dynamic hash table that
  supports an insertion in amortized $O(b^{c-1})$ I/Os and a successful
  lookup in expected average $1+O(1/b^{c})$ I/Os, or an insertion in
  amortized $\eps$ I/Os and a successful lookup in expected average
  $1+O(1/b)$ I/Os, provided that $\log\frac{n}{m} = o(b)$.
\end{theorem}

\bibliographystyle{abbrv}


\end{document}